\documentclass[runningheads]{llncs}
\usepackage{graphicx}
\usepackage{amsmath}
\usepackage{amssymb}
\usepackage{amsfonts}
\usepackage{amscd}
\usepackage{ stmaryrd } 

\usepackage[usenames]{color}
\usepackage{pgf,tikz}
\usetikzlibrary{arrows}
\usepackage{tikz-cd}  
\usepackage{hyperref}

\newcommand{\salg}[1]{\mathfrak {#1}}

\newcommand{\Ob}{\operatorname{Ob}}

\newcommand{\cat}[1]{\mathbf{#1}}

\newcommand{\MeasSets}{\cat{Meas}}

\newcommand{\diff}[2]{\frac{\mathrm{d}  #1}{\mathrm{d}  #2}}

\def \d{\mbox{\(\,\mathrm{d}\)}}
 
\renewcommand{\epsilon}{\varepsilon}

\newcommand{\norm}[1]{\left\|#1\right\|}

\newcommand{\sm}{\setminus}

\newcommand{\Rr}{\mathbb{R}}

\newcommand{\Nn}{\mathbb{N}}

\newcommand{\set}[2]{\{\,#1 \, \vert \, #2\,\} }

\usepackage{mathtools} 

\newcommand{\Sinf}{\cat{S}}
\newcommand{\Smon}{\mathcal S}
\newcommand{\Us}{\mathbf{1}}

\newcommand{\sheaf}[1]{\mathcal{#1}}

\begin{document}
\title{Information cohomology of classical vector-valued observables}
%
%
\author{Juan Pablo Vigneaux
\inst{1,2}
}
\authorrunning{J.~P. Vigneaux}
%
\institute{Institut de Math\'ematiques de Jussieu--Paris Rive Gauche (IMJ-PRG), Universit\'e de Paris, 8 place Aur\'elie N\'emours, 75013 Paris, France. \and 
Max Planck Institute for Mathematics in the Sciences, Inselstra{\ss}e 22, 04103 Leipzig, Germany. \\
\href{https://orcid.org/0000-0003-4696-4537}{orcid.org/0000-0003-4696-4537}
}
\maketitle              
\begin{abstract}
We provide here a novel algebraic characterization of two \emph{information measures} associated with  a vector-valued random variable, its differential entropy and the dimension of the underlying space, purely based on their recursive properties (the chain rule and the nullity-rank theorem, respectively). More precisely, we compute the information cohomology of Baudot and Bennequin with coefficients in a module of continuous probabilistic functionals over a category that mixes discrete observables and continuous vector-valued observables, characterizing completely the $1$-cocycles; evaluated on continuous laws, these cocycles are linear combinations of the differential entropy and the dimension. 

\end{abstract}

\section{Introduction}


 Baudot and Bennequin \cite{Baudot2015} introduced \emph{information cohomology}, and identified Shannon entropy as a nontrivial cohomology  class in degree 1. This cohomology has an explicit description in terms of \emph{cocycles} and \emph{coboundaries}; the \emph{cocycle equations} are a rule that relate different values of the cocycle. When its coefficients are a module of measurable \emph{probabilistic functionals} on a category of \emph{discrete} observables, Shannon's entropy defines a $1$-cocycle and the aforementioned rule is simply the chain rule; moreover, the cocycle equations in every degree are systems of functional equations and one can use the techniques developed by Tverberg, Lee, Knappan, Acz\'el, Dar\'oczy, etc. \cite{Aczel1975} to show that, in degree one, the entropy is the \emph{unique} measurable, nontrivial solution. The theory thus gave a new algebraic characterization of this information measure based on topos theory \emph{\`a la} Grothendieck, and showed that 
\emph{the chain rule is its defining algebraic property.}

It is natural to wonder if a similar result holds for the \emph{differential entropy}. 
We consider here information cohomology with coefficients in a presheaf of continuous probabilistic functionals on a category that mixes discrete and continuous (vector-valued) observables, and establish that every $1$-cocycle, when evaluated on probability measures absolutely continuous with respect to the Lebesgue measure, is a linear combination of the differential entropy and the dimension of the underlying space (the term continous has in this sentence three different meanings).  We already showed that this was true for gaussian measures \cite{Vigneaux2019-thesis}; in that case, there is a finite dimensional parametrization of the laws, and  we were able to use Fourier analysis to solve the $1$-cocycle equations. Here we exploit that result, expressing any density as a limit of gaussian mixtures (i.e. convex combinations of gaussian densities), and then using the $1$-cocycle condition to compute the value of a $1$-cocycle on gaussian mixtures in terms of its value on discrete laws and gaussian laws. The result depends on the conjectural existence of a ``well-behaved'' class of probabilities and probabilistic functionals, see Section \ref{sec:probas}.

The dimension appears here as an information quantity in its own right: its ``chain rule'' is the nullity-rank theorem. In retrospective, its role as an information measure is already suggested by  old results in information theory. For instance, the expansion of Kolmogorov's $\epsilon$-entropy  $H_\epsilon(\xi)$ of a continuous, $\Rr^n$-valued random variable $\xi$ ``is determined first of all by the dimension of the space, and the differential entropy $h(\xi)$ appears only in the form of the second term of the expression for $H_\epsilon(\xi)$.'' \cite[Paper 3, p. 22]{Kolmogorov1993}

\section{Some known results about information cohomology}\label{sec:known-results}

Given the severe length constraints, it is impossible to report here the motivations behind information cohomology, its relationship with traditional algebraic characterizations of entropy, and its topos-theoretic foundations. For that,  the reader is referred to the introductions of  \cite{Vigneaux2019-thesis} and \cite{Vigneaux2020information}. We simply remind here a minimum of definitions in order to make sense of the characterization of $1$-cocycles that is used later in the article. 

Let $\cat S$ be a partially ordered set (poset); we see it as a category, denoting the order relation by an arrow. It is supposed to have a terminal object $\top$ and to satisfy the following property: whenever $X,Y,Z\in \Ob \cat S$ are such that $X\to Y$ and $X\to Z$, the categorical product $Y\wedge Z$ exists in $\cat S$. An object of  $X$ of $\cat S$ (i.e. $X\in \Ob \cat S$) is interpreted as an \emph{observable}, an arrow $X\to Y$ as $Y$ being coarser than $X$, and $Y\wedge Z$  as the joint measurement of $Y$ and $Z$. 

The category $\cat S$ is just an algebraic way of encoding the relationships between observables. The measure-theoretic ``implementation'' of them comes in the form of a functor $\sheaf E:\cat S\to \cat{Meas}$ that  associates to each $X \in \Ob \cat S$ a measurable set $\sheaf E(X)=(E_X,\salg B_X)$, and to each arrow $\pi:X\to Y$ in $\cat S$ a measurable \emph{surjection} $\sheaf E(\pi):\sheaf E(X)\to \sheaf E(Y)$. To be consistent with the interpretations given above, one must suppose that   $E_\top \cong \{\ast\}$ and that $\sheaf E(Y\wedge Z)$ is mapped \emph{injectively} into $\sheaf E(Y)\times \sheaf E(Z)$ by  $\sheaf E(Y\wedge Z \to Y)\times\sheaf E(Y\wedge Z \to Z)$. 
We consider mainly two examples: the discrete case, in which $E_X$ finite and $\salg B_X$ the collection of its subsets, and the Euclidean case, in which $E_X$ is a Euclidean space and $\salg B_X$ is its Borel $\sigma$-algebra. The pair $(\cat S, \sheaf E)$ is an \emph{information structure}.

Throughout this article, conditional probabilities are understood as disintegrations. Let $\nu$ a $\sigma$-finite measure on a  measurable space $(E,\salg B)$, and $\xi$ a $\sigma$-finite measure on $(E_T,\salg B_T)$. The measure $\nu$ has a \emph{disintegration} $\{\nu_t\}_{t\in E_T}$ with respect to a measurable map $T:E\to E_T$ and $\xi$, or a $(T,\xi)$-disintegration, if each $\nu_t$ is a $\sigma$-finite measure on $\salg{B}$ concentrated on $\{T=t\}$---i.e. $\nu_t(T\neq t)=0$ for $\xi$-almost every $t$---and for each measurable nonnegative function $f:E\to \Rr$, the mapping $t\mapsto \int_E f \d \nu_t$ is measurable and $\int_E f\d \nu= \int_{E_T} \left(\int_E  f(x) \d\nu_t(x) \right)\d \xi(t)$  \cite{Chang1997}. 

We associate to each $X\in \Ob \cat S$ the set $\Pi(X)$ of probability measures on $\sheaf E(X)$ i.e. of possible \emph{laws} of $X$, and to each arrow $\pi:X\to Y$ the \emph{marginalization map} $\pi_* :=  \Pi(\pi): \Pi(X)\to  \Pi(Y)$ that maps $\rho$ to the image measure $\sheaf E(\pi)_*(\rho)$. More generally, we consider any subfunctor $\sheaf Q$ of $\Pi$ that is stable under conditioning: for all $X\in \Ob \Sinf$, $\rho\in \sheaf Q(X)$, and $\pi:X\to Y$, $\rho|_{Y=y}$ belongs to $\sheaf Q(X)$ for $\pi_*\rho$-almost every  $y\in E_Y$, where $\{\rho|_{Y=y}\}_{y\in E_Y}$ is the $(\sheaf E\pi, \pi_*\rho)$-disintegration of $\rho$.

We associate to each $X\in \Ob \cat S$ the set $\sheaf S_X=\set{Y}{X\to Y}$, which is a monoid under the product $\wedge$ introduced above. The assignment $X\mapsto \sheaf S_X$ defines a contravariant functor (presheaf). The induced algebras $\mathcal A_X= \Rr[\sheaf S_X]$ give a presheaf $\sheaf A$. An $\sheaf A$-module is a collection of modules $\sheaf M_X$ over $\sheaf A_X$, for each $X\in \Ob \cat S$, with an action that is ``natural'' in $X$.  The main example is the following: for any adapted probability functor $\sheaf Q:\cat S \to \cat{Meas}$, one introduces a \emph{contravariant} functor $\sheaf F = \sheaf F(\sheaf Q)$ declaring that $\sheaf F(X)$ are the measurable functions on $\sheaf Q(X)$, and $\sheaf F(\pi)$ is precomposition with $\sheaf Q(\pi)$ for each morphism $\pi$ in $\cat S$. The monoid $\sheaf S_X$ acts on $ \sheaf F(X)$ by the rule:
\begin{equation}\label{eq:conditional_action}
\forall Y\in \Smon_X, \forall \rho\in \sheaf Q(X), \quad Y.\phi(\rho) = \int_{E_Y} \phi(\rho|_{Y=y})\d \pi^{YX}_*\rho(y)
\end{equation}
where $\pi^{YX}_*$ stands for the marginalization $\sheaf Q(\pi^{YX})$ induced by $\pi^{YX}:X\to Y$ in $\cat S$. This action can be extended by linearity to $\sheaf A_X$ and is natural in $X$. 

In \cite{Vigneaux2020information}, the information cohomology $H^\bullet(\cat S, \sheaf F)$ is defined using \emph{derived functors} in the category of $\mathcal A$-modules, and then described explicitly, for each degree $n\geq 0$, as a quotient of $n$-cocycles by $n$-coboundaries. For $n=1$, the coboundaries vanish, so we simply have to describe the cocycles. Let $\sheaf B_1(X)$ be the $\sheaf A_X$-module freely generated by a collection of bracketed symbols $\{[Y]\}_{Y\in \sheaf S_X}$; an arrow $\pi:X\to Y$ induces an inclusion $\sheaf B_1(Y)\hookrightarrow \sheaf B_1(X)$, so $\sheaf B_1$ is a presheaf. A \emph{$1$-cochain} is a natural transformations $\varphi: \sheaf B_1 \Rightarrow \sheaf F$, with components $\varphi_X:\sheaf B_1(X) \to \sheaf F(X)$; we use $\varphi_X[Y]$ as a shorthand for $\varphi_X([Y])$. The naturality implies that $\varphi_X[Z](\rho)$ equals $\varphi_Z[Z](\pi^{ZX}_*\rho)$, 
 a property that \cite{Baudot2015} called \emph{locality}; sometimes we write $\Phi_Z$ instead of  $\varphi_Z[Z]$. A $1$-cochain $\varphi$ is a $1$-cocycle iff
\begin{equation}\label{eq:1-cocycle-cond}
\forall X \in \Ob \Sinf, \forall X_1,X_2\in \sheaf S_X,  \quad \varphi_X[X_1\wedge X_2]  = X_1.\varphi_X[X_2]+ \varphi_{X}[X_1]. 
\end{equation}
Remark that this is an equality \emph{of functions} in $\sheaf Q(X)$. 

An information structure is \emph{finite} if for all $X\in \Ob S$, $E_X$ is finite. In this case, \cite[Prop.~4.5.7]{Vigneaux2020information} shows that, whenever an object $X$ can be written as a product $Y\wedge Z$ and $E_X$ is ``close'' to $E_Y\times E_Z$, as formalized by the definition of \emph{nondegenerate product} \cite[Def.~4.5.6]{Vigneaux2020information}, then there exists $K\in \Rr$ such that  for all $W\in \sheaf S_X$ and $\rho$ in $\sheaf Q(Z)$
\begin{equation}\label{discrete_1_cocycle}
  \Phi_W(\rho) = -K\sum_{w\in E_W} \rho(w) \log \rho(w).
\end{equation}

The continuous case is of course more delicate. In the case of $\sheaf E$ taking values in vector spaces, and $\sheaf Q$ made of gaussian laws, \cite{Vigneaux2019-thesis} treated it as follows. We start with a vector space $E$ with Euclidean metric $M$, and a poset $\Sinf$ of vector subspaces of $E$, ordered by inclusion, satisfying the hypotheses stated above; remark that $\wedge$ corresponds to intersection. Then we introduce $\sheaf E$ by $V\in \Ob \Sinf\mapsto E_V:=E/V$, and further identify $E_V$ is $V^\perp$ using the metric (so that we only deal with vector subspaces of $E$).  We also introduce a sheaf $\sheaf N$, such that $\sheaf N(X)$ consists of affine subspaces of $E_X$ and is closed under intersections; the sheaf is supposed to be closed under the projections induced by $\sheaf E$  and to contain the fibers of all these projections. On each affine subspace $N\in \sheaf N(X)$ there is a unique Lebesgue measure $\mu_{X,N}$ induced by the metric $M$. We consider a sheaf $\sheaf Q$ such that $\sheaf Q(X)$ are probabilities measures $\rho$ that are absolutely continuous with respect to $\mu_{X,N}$ for some $N\in \sheaf N(X)$ and have a gaussian density with respect to it. We also introduce a subfunctor $\sheaf F'$ of $\sheaf F$ made of functions that \emph{grow moderately} (i.e. at most polynomially) with respect to the mean, in such a way that the integral \eqref{eq:conditional_action} is always convergent. Ref. \cite{Vigneaux2019-thesis} called a triple $(\cat S, \sheaf E, \sheaf N)$ \emph{sufficiently rich} when there are ``enough supports'', in the sense that one can perform marginalization and conditioning with respect to projections on subspaces generated by elements of at least two different bases of $E$. In this case, we showed that every $1$-cocycle $\varphi$, with coefficients in $\sheaf F'(\sheaf Q)$, there are real constants $a$ and $c$ such that, for every $X \in \Ob \Sinf$ and every gaussian law $\rho$ with support $E_X$ and variance $\Sigma_{\rho}$ (a nondegenerate, symmetric, positive bilinear form on $E_X^*$),
\begin{equation}\label{cont_1_cocycle}
\quad\Phi_X(\rho) = a \det(\Sigma_{\rho}) + c.\dim(E_X).
\end{equation}
 Moreover, $\varphi$ its completely determined by its behavior on nondegenerate laws. (The measure $\mu_X=\mu_{X,E_X}$ is enough to define the determinant $\det(\Sigma_{\rho})$ \cite[Sec.~11.2.1]{Vigneaux2019-thesis}, but the latter can also be computed w.r.t. a basis of $E_X$ such that $M|_{E_X}$ is represented by the identity matrix.)

\section{An extended model}

\subsection{Information structure, supports, and reference measures} 

In this section, we introduce a more general model, that allows us to ``mix'' discrete and continuous variables. It is simply the product of a structure of discrete observables and a structure of continuous ones.

 Let $(\Sinf_d, \sheaf E_d)$ be a finite information structure, such that  for every $n\in \Nn$, there exist $Y_n\in \Ob \Sinf_d$ with $|E_{Y_n}|\cong \{1,2,...,n\}=:[n]$, and for every $X\in \Ob\Sinf_d$, there is a $Z\in \Ob \Sinf_d$ that can be written as non-degenerate product and such that $Z\to X$; this implies that \eqref{discrete_1_cocycle} holds for every $W\in \Ob \Sinf_d$. Let $(\Sinf_c, \sheaf E_c, \sheaf N_c)$ be a \emph{sufficiently rich triple} in the sense of the previous section, associated to a vector space $E$ with metric $M$, so that \eqref{cont_1_cocycle} holds for every $X\in \Ob \Sinf_c$. 
 
Let  $(\Sinf,\sheaf E:\Sinf \to \MeasSets)$ be the product $(\Sinf_c, \sheaf E_c)\times(\Sinf_d, \sheaf E_d)$ in the category of information structures, see \cite[Prop.~2.2.2]{Vigneaux2020information}. By definition, every object $X\in \Sinf$ has the form $\langle X_c,  X_d\rangle$ for $X_c\in \Ob \Sinf_c$ and $X_d\in\Ob \Sinf_d$, and $\sheaf E(X) = \sheaf E(X_1) \times \sheaf E(X_2)$, and there is an arrow $\pi:\langle X_1,X_2\rangle \to \langle Y_1,Y_2\rangle$ in $\Sinf$ if and only if there exist arrows $\pi_1:X_1\to Y_1$ in $\Sinf_c$ and $\pi_2:X_2\to Y_2$ in $\Sinf_d$; under the functor $\sheaf E$, such $\pi$ is mapped to $\sheaf E_c(\pi_1)\times \sheaf E_d(\pi_2)$.   There is an embedding in the sense of information structures, see \cite[Sec.~1.4]{Vigneaux2019-thesis}, $\Sinf_c \hookrightarrow \Sinf$, $X\to \langle X, \Us\rangle$; we call its image the ``continuous sector'' of $\Sinf$; we write $X$ instead of $\langle X, \Us \rangle$ and $\sheaf E(X) $ instead of $\sheaf E (\langle X, \Us \rangle) = E(X) \times \{\ast\}$. The ``discrete sector'' is defined analogously.

We extend the sheaf of supports $\sheaf N_c$ to the whole $\Sinf$ setting $\sheaf N_d(Y) = 2^{Y}\sm \{\emptyset\}$ when $Y\in \Ob \Sinf_d$, and  then $\sheaf N(Z) = \set{ A \times B }{ A \in N_c(X), B\in N_d(Y)}$ for any $Z=\langle X, Y \rangle\in \Ob \Sinf$. The resulting   $\sheaf N$ is a functor on $\Sinf$ closed under projections and intersections, that contains the fibers of every projection.  

For every $X\in \Ob \Sinf$ and $N\in \sheaf N(X)$, there is a unique reference measure $\mu_{X,N}$ compatible with $M$: on the continuous sector it is the Lebesgue measure given by the metric $M$ on the affine subspaces of $E$, on the discrete sector  it is the counting measure restricted to $N$, and for any product $A\times B \subset E_X\times E_Y$, with $X\in \Ob \Sinf_c$ and $Y\in \Ob \Sinf_d$, it is just $\sum_{y\in B} \mu_A^y$, where $\mu_A^y$ is the image of $\mu_A$ under the inclusion $A\hookrightarrow A\times B, \; a\mapsto (a,y)$. We write $\mu_X$ instead of $\mu_{X, E_X}$.
%

The disintegration of the counting measure into counting measures is trivial. The disintegration of the  Lebesgue measure $\mu_{V,N}$ on a support $N\subset E_V$ under the projection $\pi^{WV}:E_V\to E_W$ of vector spaces is given by the Lebesgue measures on the fibers $(\pi^{W,V})^{-1}(w)$, for $w\in E_W$. We recall that we are in a framework where $E_V$ and $E_W$ are identified with subspaces of $E$, which has a metric $M$; the disintegration formula is just Fubini's theorem.  

To see that disintegrations exist under any arrow of the category, consider first an object 
 $Z= \langle X, Y \rangle$ and  arrows  $\tau:Z \to X$ and $\tau':Z\to Y$, when $X\in \Ob \Sinf_c$ and $Y\in \Ob \Sinf_d$. By definition $E_Z = E_Y \times E_X = \cup_{y\in E_Y}  E_X \times  \{y\}$, and the canonical projections $\pi^{YZ}:E_Z\to E_Y$ and $\pi^{XZ}:E_Z\to E_X$ are the images under $\sheaf E$ of $\tau$ and $\tau'$, respectively. Set $E_{Z,y} :=  (\pi^{YZ})^{-1}(y)= E_X \times \{y\}$ and $E_{Z,x} :=  (\pi^{XZ})^{-1}(x)= \{x\}\times E_Y$. According to the previous definitions,  $E_Z$ has reference measure 
$
\mu_Z = \sum_{y\in E_Y} \mu_X^{y},
$ where $\mu_X^{y}$ is the image of $\mu_X$ under the  inclusion $E_X\to E_X \times \{y\}$. Hence by definition, $\{\mu_X^{y}\}_{y\in E_Y}$ is $(\pi^{YZ},\mu_Y)$-disintegration of $\mu_Z$. Similarly, $\mu_Z$ has as $(\pi^{XZ},\mu_X)$-disintegration the family of measures $\{\mu_Y^{x}\}$, where each $\mu_Y^{x}$ is the counting measure restricted on the fiber $E_{Z,x}\cong E_Y$. 

More generally, the disintegration of reference measure $\mu_{Z,A\times B}=  \sum_{y\in B} \mu_A^y$ on a support $A\times B$ of $Z=\langle X , Y\rangle$ under the arrow $\langle \pi_1,\pi_2\rangle :Z\to Z' = \langle X' , Y'\rangle$ is the collection of measures $ (\mu_{Z,A\times B,x',y'})_{(x',y')\in \pi_1(A)\times \pi_2(B)}$ such that 
\begin{equation}
\mu_{Z,A\times B,x',y'} = \sum_{y\in \pi_2^{-1}(y')} \mu^{y}_{X,A,x'}
\end{equation}
 where $\mu^{y}_{X,A,x'}$ is the image measure, under the inclusion $E_X \to E_X\times \{y\}$, of the measure $\mu_{X,A,x'}$ that comes from the $(\pi_1,\mu_{X',\pi_1(A)})$-disintegration of $\mu_{X,A}$. 
 

\subsection{Probability laws and probabilistic functionals}\label{sec:probas}

Consider the subfunctor $\Pi(\sheaf N)$ of $\Pi$ that associates to each $X\in \Ob \Sinf$ the set $\Pi(X;\sheaf N)$ of probability measures on $\sheaf E(X)$ that are absolutely continuous with respect to the reference measure $\mu_{X,N}$  on some $N\in \sheaf N(X)$. We define the \emph{(affine) support} or \emph{carrier} of $\rho$, denoted $A(\rho)$, as  the unique $A\in \sheaf N(X)$ such that $\rho\ll \mu_{X,A}$.

In this work, we want to restrict our attention to subfunctors $\sheaf Q\subset \Pi(\sheaf N)$ of probability laws such that:
\begin{enumerate}
\item\label{conditionQ:adapted} $\sheaf Q$ is adapted;
\item\label{conditionQ:existence} for each $\rho\in \sheaf Q(X)$, the differential entropy $S_{\mu_{A(\rho)}}(\rho):=- \int \log \diff{\rho}{\mu_{A(\rho)}}\d \rho$ exists i.e. it is finite;
\item\label{conditionQ:continuity} when restricted to probabilities in $\sheaf Q(X)$ with the same carrier $A$, the differential entropy is a continuous functional in the total variation norm;
\item \label{conditionQ:gaussians} for each $X\in \Ob\cat S_c$ and each $N\in \sheaf N(X)$, the \emph{gaussian mixtures} carried by $N$ are contained in $\sheaf Q(X)$---cf. next section.
\end{enumerate}

\begin{problem}\label{problem1:class_probas}
The characterization of functors $\sheaf Q$ that satisfy properties \ref{conditionQ:adapted}.-\ref{conditionQ:gaussians}.
\end{problem}

Below we use kernel estimates, which interact nicely with the total variation norm. This norm is defined for every measure $\rho$ on $(E_X,\salg B_X)$ by $\norm{\rho}_{TV} =  \sup_{A\in \salg B_X} |\rho(A)|$. Let  $\varphi$ be a real-valued functional defined on  $\Pi(E_X,\mu_A)$, and $L^1_1(A,\mu_A)$ the space of functions $f\in L^1(A,\mu_A)$ with total mass 1 i.e. $\int_A f\d \mu_A = 1$;  the continuity of $\varphi$  in the total variation distance is equivalent to the continuity of $\tilde \varphi:L_1^1(A,\mu_A)\to \Rr, f\mapsto\varphi(f \cdot \mu_A)$ in the $L^1$-norm, because of Scheffe's identity \cite[Thm.~45.4]{Devroye2002}.

  The characterization referred to in Problem \ref{problem1:class_probas} might involve the densities or the moments of the laws. It is the case with the main result that we found concerning convergence of the differential entropy and its continuity in total variation \cite{Ghourchian2017}. Or it might resemble  Ot\'ahal's result \cite{Otahal1994}: if densities $\{f_n\}$ tend to $f$ in $L^\alpha(\Rr,\d x)$ and $L^\beta(\Rr,\d x)$, for some  $0<\alpha < 1 <\beta$, then $-\int f_n(x) \log f_n(x) \d x\to - \int f(x) \log f(x) \d x$.   

For each $X\in \Ob\Sinf$, let  $\sheaf F(X)$ be the vector space of measurable functions of  $(\rho, \mu_{M})$, equivalently $(\d\rho/d\mu_{A(\rho)}, \mu_{M})$, where $\rho$ is an element of $\sheaf Q(X)$, $\mu_{M}$ is a global determination of reference measure on any affine subspace given by the metric $M$, and $\mu_{A(\rho)}$ is the corresponding reference measure on the carrier of $\rho$ under this determination. 

We want to restrict our attention to functionals for which the action \eqref{eq:conditional_action}
is integrable. Of course, these depends on the answer to  Problem \ref{problem1:class_probas}. 

\begin{problem}\label{problem2:moderate_functionals}
What are the appropriate restrictions on the functionals $\sheaf F(X)$ to guarantee the convergence of \eqref{eq:conditional_action}?
\end{problem}

\section{Computation of $1$-cocycles}\label{sec:computations_1cocycle}

\subsection{A formula for gaussian mixtures}
Let $\mathcal Q$ be a probability functor satisfying conditions \eqref{conditionQ:adapted}-\eqref{conditionQ:gaussians} in Subsection \ref{sec:probas}, and $\sheaf G$ a linear subfunctor of $\sheaf F$ such that \eqref{eq:conditional_action} converges for laws in $\sheaf Q$. In this section, we compute $H^1(\cat S, \sheaf G)$. 

Consider a generic object $Z= \langle X , Y\rangle = \langle X , 1 \rangle \wedge \langle 1 , Y\rangle$ of $\Sinf$; we write everywhere $X$ and $Y$ instead of $\langle X , 1 \rangle$ and $\langle 1 , Y\rangle$. We suppose that $E_X$  is an   Euclidean space of dimension $d$. Remind that  $E_Y$ is a finite set and $E_Z = E_X\times E_Y$. Let $\{G_{M_y, \Sigma_y}\}_{y\in E_Y}$ be gaussian densities on $E_X$ (with mean $M_y$ and covariance $\Sigma_y$), and $p:E_Y\to [0,1]$ a density on $E_Y$; then $\rho=\sum_{y\in E_Y} p(y) G_{M_y, \Sigma_y} \mu_X^{y}$ is a probability measure on $E_Z$, absolutely continuous with respect to $\mu_Z$, with density $r(x,y) :=  p(y) G_{M_y, \Sigma_y}(x)$. We have that $\pi^{YZ}_*\rho$ has density $p$ with respect to the counting measure $\mu_Y$, whereas  $\pi^{XZ}_*\rho$ is absolutely continuous with respect to $\mu_X$ (see \cite[Thm.~3]{Chang1997}) with density
\begin{equation}
\mu_Z^x(r) = \int_{E_X} \diff \rho {\mu_Z} \d\mu_Z^{x} = \sum_{y\in E_Y} p(y) G_{M_y, \Sigma_y}(x).
\end{equation}
i.e. it is a \emph{gaussian mixture}. For conciseness, we utilize here  linear-functional notation for some integrals, e.g. $\mu_Z^x(r)$. 
The measure $\rho$ has a $\pi^{XZ}$-disintegration into probability laws $\{\rho_x\}_{x\in E_X}$, such that each $\rho_x$ is concentrated on $E_{Z,x}\cong E_Y$ and 
\begin{equation}
\rho_x(x,y) = \frac{p(y)  G_{M_y, \Sigma_y}(x)}{\sum_{y\in E_Y} p(y) G_{M_y, \Sigma_y}(x)}.
\end{equation}
In virtue of the cocycle condition (remind that $\Phi_Z = \varphi_Z[Z]$),
\begin{equation}
\Phi_Z(\rho) = \varphi_Z[Y](\rho) + \sum_{y\in E_Y} p(y) \varphi_Z[X](G_{M_y, \Sigma_y}\mu_X^{y}) 
\end{equation}
Locality implies that $\varphi_Z[Y](\rho) =  \Phi_{Y}(\pi^{YZ}_*\rho)$, and the later equals \linebreak $-b \sum_{y\in E_Y} p(y) \log p(y)$, for some $b\in \Rr$, in virtue of the characterization of cocycles restricted to the discrete sector. Similarly, $\varphi_Z[X](G_{M_y, \Sigma_y}\mu_X^{y}) = \Phi_X(G_{M_y, \Sigma_y}) = a \log \det(\Sigma_y)  + c \dim E_X$ for some $a,c\in \Rr$, since our hypotheses are enough to characterize the value of any cocycle restricted to the gaussian laws on the continuous sector. Hence
\begin{equation}
\Phi_Z(\rho) =
 -b \sum_{y\in E_Y} p(y) \log p(y) +  \sum_{y\in E_Y} p(y) (a \log \det(\Sigma_y)  + c \dim E_X). \label{first_formula_global_rho}
 \end{equation}

Remark that the same argument can be applied to any densities $\{f_y\}_{y\in Y}$ instead $\{G_{M_y, \Sigma_y}\}$. Thus, it is enough to determine $\Phi_{X}$  on general densities, for each $X\in \Sinf_c$, to characterize completely the cocycle $\varphi$.

The cocycle condition also implies that 
\begin{align}
\Phi_Z(\rho) &= \Phi_X(\pi^{XZ}\rho) + \int_{E_X} \varphi_Z[Y](\rho_x) \d \pi^{XZ}_*\rho(x). \label{second_formula_global_rho}
\end{align}
The law $\pi^{XZ}_*\rho$ is a composite gaussian, and the law $\rho_x$ is supported on the discrete space $E_{Z,x}\cong E_Y$, with density 
\begin{equation}
(x,y) \mapsto \rho_x(y)= \frac{p(y)G_{M_y,\Sigma_y}(x)}{\sum_{y'\in E_Y} p(y')G_{M_y,\Sigma_y}(x)} = \frac{r(x,y)}{\mu_Z^{x}(r)}.
\end{equation}
Using again locality and the characterization of cocycles on the discrete sector, we deduce that
\begin{equation}
\varphi_Z[Y](\rho_x)= \Phi_Y(\pi^{YZ}_*\rho_x) = -b\sum_{y\in E_Y} \rho_x(y) \log \rho_x(y).
\end{equation} 
A direct computation shows that $\sum_{y\in E_Y} \rho_x(y) \log \rho_x(y)$ equals:
\begin{equation}
\sum_{y\in E_Y} p(y) \log p(y) - 
 \sum_{y\in E_Y} p(y) S_{\mu_X} (G_{M_y,\Sigma_y})   + S_{\mu_X}(\mu_Z^{x}(r)),
\end{equation}
where $S_{\mu_X}$ is the differential entropy, defined for any $f\in L^1(E_X,\mu_X)$ by $S_{\mu_X}(f) =  -\int_{E_X} f(x) \log f(x) \d \mu_X(x)$. 

It is well known that  $ S_{\mu_X}(G_{M,\Sigma}) = \frac 12 \log \det \Sigma + \frac{\dim{E_X}}{2} \log(2\pi e).$ 

Equating the right hand sides of \eqref{first_formula_global_rho} and \eqref{second_formula_global_rho}, we conclude that 
\begin{equation}\label{eq:entropy_composite_normal_bis}
\Phi_X(\pi^{XZ}\rho) = \sum_{y\in E_Y} p(y) \left((a -\frac b2)\log\det\Sigma_y + cd - \frac{bd}{2} \log(2\pi e)) \right)  + b S_{\mu_X} (\mu_Z^{x}(r)) .
\end{equation}

%


\subsection{Kernel estimates and main result}
Equation \ref{eq:entropy_composite_normal_bis} gives an explicit value to the functional $\Phi_X\in \sheaf F(X)$ evaluated on a gaussian mixture. Any density in $L^1(E_X,\mu_X)$ can be approximated by a random mixture of gaussians. This approximation is known as a kernel estimate. 

Let $X_1, X_2,...$ be a sequence of independently distributed random elements of $\Rr^d$, all having  a common density $f$ with respect to the  Lebesgue measure $\lambda_d$. 
Let $K$ be a nonnegative Borel measurable function, called \emph{kernel}, such that  $\int K \d \lambda_d=1$, and $(h_n)_n$ a sequence of positive real numbers. The \emph{kernel estimate} of $f$ is given by 
\begin{equation}\label{eq:kernel_estimate}
f_n(x) = \frac{1}{nh_n^d} \sum_{i=1}^n K \left(\frac{x-X_i}{h_n}\right).
\end{equation}
The distance $J_n = \int_E |f_n-f|\d\lambda_d$ is a random variable, invariant under arbitrary automorphisms of $E$. The key result concerning these estimates \cite[Ch.~3, Thm.~1]{Devroye1985} says, among other things, that $J_n \to 0$ in probability as $n\to \infty$ for \emph{some} $f$ if and only if $J_n \to 0$ almost surely  as $n\to \infty$ for \emph{all} $f$, which holds if and only if $\lim_n h_n = 0$ and $\lim_n nh_n^d =\infty$.

\begin{theorem}
Let $\varphi$ be a $1$-cocycle on $\Sinf$ with coefficients in $\mathcal G$, $X$ an object in $\Sinf_c$, and $\rho$ a probability law in $\sheaf Q(X)$ absolutely continuous with respect to $\mu_X$. Then, there exist real constants $c_1, c_2$ such that
\begin{equation}
\Phi_X(\rho):=\varphi_X[X](\rho) = c_1 S_{\mu_A} (\rho) + c_2 \dim E_X.
\end{equation}
\end{theorem}
\begin{proof}
By hypothesis. $X=\langle X, \Us\rangle$ is an object in the continuous sector of $\Sinf$. Let  $f$ be any density of $\rho$ with respect to $\mu_X$,  $(X_n)_{n\in \Nn}$ an i.i.d sequence  of points of $E_X$ with law $\rho$, and  $(h_n)$  any sequence such that $h_n\to 0$ and $nh_n^d \to \infty$. Let $(X_n(\omega))_{n}$ be any realization of the process such that $f_n$ tend to $f$ in $L^1$. We introduce, for each $n\in \Nn$, the kernel estimate \eqref{eq:kernel_estimate} evaluated at $(X_n(\omega))_n$, taking $K$ equal to the density of a standard gaussian. Each $f_n$ is the density of a composite gaussian law $\rho_n$ equal to $n^{-1} \sum_{i=1}^n G_{X_i(\omega), h_n^2 I}$. This can be ``lifted'' to $Z = \langle X, Y_n\rangle$, this is, there exists a law $\tilde \rho$ on $E_Z:=E_X\times [n]$ with density $r(x,i) = p(i) G_{X_i(\omega), h_n^2I} (x)$, where $p:[n]\to \Rr$ is taken to be the uniform law. The arguments of Section \ref{sec:computations_1cocycle} then imply that
\begin{equation}\label{eq:entropy_composite_normal_bis}
\Phi_X(\rho_n) =  2d\left(a - \frac b2 \right) \log h_n + cd - \frac{bd}{2} \log(2\pi e)  + b S_{\mu_A} (\rho_n).
\end{equation}
In virtue of the hypotheses on $\sheaf Q$, $S_{\mu_A}(\rho)$  is finite and $S_{\mu_A} (\rho_n) \to S_{\mu_A}(\rho)$. Since $\Phi_X$ is continuous when restricted to $\Pi(A,\mu_A)$ and $\Phi_X(f)$ is a real number, we conclude that necessarily $a=b/2$. The statement is then just a rewriting of \eqref{eq:entropy_composite_normal_bis}.
\end{proof}

This is the best possible result: the dimension is an invariant associated to the reference measure, and the entropy depends on the density.

\bibliographystyle{splncs04}
\bibliography{biblio}

\end{document}